\documentclass[11pt]{article}
\usepackage{amsfonts}
\usepackage{amssymb}
\usepackage{amstext}
\usepackage{amsmath}
\usepackage{xspace}
\usepackage{theorem}
\usepackage{graphicx}
\usepackage{graphics}
\usepackage{colordvi}
\usepackage{color}
\usepackage{hyperref}
\usepackage{paralist}
\usepackage{bm}
\usepackage{mathpazo}
\usepackage{algpseudocode} 

\newcommand{\CC}[0]{\operatorname{Constrained-Cut}}
\algnewcommand{\InlineComment}[1]{\hfill {\itshape // #1}}

\advance \oddsidemargin by -0.8in

        {\hspace*{\fill}$\Box$\par\vspace{4mm}}

\newcommand{\set}[1]{\left\{ #1 \right\}}

\newcommand{\parenbig}[1]{\bigl( #1 \bigr)}

\newcommand{\Hcal}[0]{\ensuremath{\mathcal H}\xspace} 
\newcommand{\Tcal}[0]{\ensuremath{\mathcal T}\xspace} 
\newcommand{\Gnet}[0]{\ensuremath{\Hcal'}\xspace}
\newcommand{\Dnet}[0]{\ensuremath{\Hcal}\xspace}

\DeclareMathOperator{\tw}{tw}
\DeclareMathOperator{\rootn}{root}

\newcommand{\be}{\begin{enumerate}}
\newcommand{\ee}{\end{enumerate}}
\newcommand{\bd}{\begin{description}}
\newcommand{\ed}{\end{description}}
\newcommand{\bi}{\begin{itemize}}
\newcommand{\ei}{\end{itemize}}

\newtheorem{theorem}{Theorem}
\newtheorem{lemma}[theorem]{Lemma}
\newtheorem{observation}[theorem]{Observation}

\newtheorem{claim}[theorem]{Claim}
\newtheorem{proposition}[theorem]{Proposition}

\newtheorem{definition}[theorem]{Definition}

\newenvironment{proof}{\par \smallskip{\bf Proof:}}{\hfill\stopproof}
\def\stopproof{\square}
\def\square{\vbox{\hrule height.2pt\hbox{\vrule width.2pt height5pt \kern5pt
\vrule width.2pt} \hrule height.2pt}}

\newenvironment{prog}[1]{
\begin{minipage}{5.8 in}
{\sc\bf #1}
\begin{enumerate}}
{
\end{enumerate}
\end{minipage}
}

\renewcommand{\phi}{\varphi}

\newcommand{\poly}{\operatorname{poly}}

\newcommand{\Z}{\ensuremath{\mathbb Z}}

\newcommand{\cset}{{\mathcal C}}

\newcommand{\pset}{{\mathcal P}}
\newcommand{\tset}{{\mathcal T}}

\newcommand{\mincut}[0]{\operatorname{mincut}}
\newcommand{\Xset}{\ensuremath{X}\xspace} 
\newcommand{\card}[1]{\left| #1 \right|}
\newcommand{\nosolution}{\textsc{No Valid Solution}\xspace}

\title{Mimicking Networks Parameterized by Connectivity}

\author{Parinya Chalermsook\thanks{Aalto University, Finland.} \and 
	Syamantak Das\thanks{Indraprastha Institute of Information Technology Delhi, India} \and 
	Bundit Laekhanukit\thanks{Shanghai University of Finance and Economics, China.} \and 
	Daniel Vaz\thanks{Technische Universität München, Germany. Work done while at Max-Planck-Institut für Informatik, Germany.}} 
\begin{document}

\maketitle

\begin{abstract}
Given a graph $G=(V,E)$, capacities $w(e)$ on edges, and a subset of terminals $\tset \subseteq V: |\tset| = k$, a mimicking network for $(G,\tset)$ is a graph $(H,w')$ that contains copies of $\tset$ and preserves the value of minimum cuts separating any subset $A, B \subseteq \tset$ of terminals. 
Mimicking networks of size $2^{2^k}$ are known to exist and can be constructed algorithmically, while the best known lower bound is $2^{\Omega(k)}$; therefore, an exponential size is required if one aims at preserving cuts exactly.

In this paper, we study mimicking networks that preserve connectivity of the graph exactly up to the value of $c$, where $c$ is a parameter. This notion of mimicking network is sufficient for some applications, as we will elaborate.  
We first show that a mimicking of size $3^c \cdot k$  exists, that is, we can preserve cuts with small capacity using a network of size linear in $k$. 
Next, we show an algorithm that finds such a mimicking network in time $2^{O(c^2)} \poly(m)$.
\end{abstract}

\section{Introduction} 

\textit{Graph compression} is a basic information-theoretic and computational question of the following nature: Given an $n$-node graph $G$ (imagine $n$ to be very large), can we compute a ``compact'' (much smaller) representation of $G$ that preserves information that is important for us? 
Many such objects have been central in algorithm designs. 
For instance, graph spanners aim at preserving the distance between nodes in the graphs, while vertex sparsifiers focus on preserving the cut sizes among designated nodes. 

In this paper, we focus on \emph{mimicking networks}: Given a graph $G=(V,E)$, capacities $w(e)$, and a subset of terminals $\tset\subseteq V: |\tset| = k$, our goal is to find a smaller capacitated graph $(H,w')$ that contains copies of $\tset$ and preserves the value of minimum cuts separating any subset $A,B \subseteq \tset$. In this case, we say that $H$ is a mimicking network of $G$. 
This question was introduced by Hagerup et al.~\cite{hagerup1998characterizing} where they presented a mimicking network of size $2^{2^k}$ which depends only on $k$ but not $n$ (see also an improvement by Khan and Raghavendra~\cite{khan2014mimicking}). 
Krauthgamer and Rika~\cite{krauthgamer2013mimicking} showed that the exponential dependence on $k$ is needed; they presented a lower bound of $2^{\Omega(k)}$. 
It remains an intriguing open problem to close this gap.

While for some applications (e.g. for cut and flow problems), it is desirable to preserve the cut values exactly for every cut, this is not the case for connectivity problems. 
For instance, if we want to keep only information about $c$-connected subgraphs, it would be enough to treat all cuts of size larger than $c$ as having value exactly $c$. 
Therefore, we initiate the study of \textit{connectivity-$c$ mimicking network} and present some application in fast computation in graphs of low treewidth. 
The following is our main theorem. 

\begin{theorem}
\label{thm:connc:mimnet}
There is a connectivity-$c$ mimicking network of size $3^c \cdot c \cdot k$.
Moreover, there exists a deterministic algorithm to find such a network in
time $O\parenbig{2^{O(c^2)}\cdot k^2\cdot m^2}$.
\end{theorem}

Our main theorem shows that the exponential lower bound in $k$ does not apply in the ``low-connectivity'' setting.

\paragraph{Related work:} For special graph classes, better bounds are known. For instance, Krauthgamer and Rika~\cite{krauthgamer2013mimicking} presented an upper bound of $k^2 2^{2k}$ for planar graphs, and this was proved to be almost tight~\cite{karpov2018exponential}. 
When all terminals lie on the same face, the exponential lower bound does not apply and mimicking networks of size $O(k^2)$ are known~\cite{goranci2017improved}.
For bounded treewidth graphs, a the upper bound of $O(k)$ is known~\cite{chaudhuri2000computing}.

If we allow graph $H$ to only approximate the cut size, such an object is known under the name of cut sparsifiers (introduced in~\cite{moitra2009approximation,leighton2010extensions}): For $q \geq 1$, a quality-$q$ cut sparsifier preserves cut values between terminals up to a factor of $q$. (in this language, a mimicking network is simply a quality-$1$ cut sparsifier.)
Please refer to~\cite{englert2014vertex,andoni2014towards} and references therein for discussions on cut sparsifiers and most recent results.

\section{Preliminaries} 

For simplicity,  we view any capacitated graph $(G,w)$ as a multi-graph $G$ obtained by making $\min(w(e),c)$ copies of parallel edges of $e$.
Hence, we will only be dealing with uncapacitated,  multi-graphs from now on.   
For any graph $G$ and two disjoint subsets $A,B \subseteq V(G)$, denote by $E_G(A,B)$ the edges with one endpoint in $A$ and the other one in $B$, while, by $\mincut_G(A,B)$, the value of the minimum cut separating the sets $A$ from $B$.
If either set is empty, the mincut has a value of $0$.
We also need the notion of thresholded minimum cut; that is, for an integer $C$, denote by $\mincut_G^c(A,B) = \min \{\mincut_G(A,B), c \}$.

\paragraph{Bounded-Connectivity Mimicking Network:} For any graph $G$ and terminal set $\tset$, we say that graph $H$ is a \emph{connectivity-$c$ mimicking network} for $G$ if the following holds:
\begin{itemize} 
\item $V(H)$ contains at least one copy of each terminal in $\tset$. 
\item For any pair of disjoint subsets $A, B \subseteq \tset$ of terminals, the thresholded minimum cuts are preserved in $H$, i.e.  
\[ \mincut_H^c(A, B) = \mincut^c_G(A,B) \] 
\end{itemize}

We will assume that every terminal has degree exactly one. 
This assumption can be made by creating $c$ auxiliary vertices for each terminal $t \in \tset$, and connecting them to $t$; each of these auxiliary vertices becomes a new terminal. 
Notice that this increases the number of terminals by a factor of $c$, and the
bounds for the size of the connectivity-$c$ mimicking network correspondingly. %
From now on, we assume an input  graph $G= (S \cup \tset, E)$ where each node in $\tset$ has degree one and is attached to the set $S$. 
We interchangeably refer to terminals as either (1) the nodes in $\tset$ or (2) the edges connecting $\tset$ to $S$.

For a set of vertices $X \subseteq S$, we use the notation $\partial(X)$ to denote the set of boundary edges $E_G(X, V(G)-X)$. 
By definition, $\partial(S)= \tset$. 
Notice that for $X \subseteq S$, it could be that $\partial(X)$ contains edges that are not in $\partial(S)$.

\paragraph{Contraction-based mimicking network:} We are interested in mimicking networks with a specific structure.
Let $X \subseteq S$. Denote by $G/X$ the graph $G$ obtained after contracting every edge in $G[X]$. 
In particular, given an input $G=(S \cup \partial(S), E)$, our mimicking network is always obtained by contracting disjoint subsets $S_1,\ldots, S_{\ell} \subseteq S$.

\paragraph{Well-linked sets:} A standard tool for studying flows and cuts is the notion of well-linkedness.
We extend this notion so as to capture our bounded connectivity setting.   
A set $X \subseteq S$ is said to be a \emph{connectivity-$c$ linked set} in $G$ if for every pair of disjoint sets $A, B \subseteq X$, we have that:  
\[|E_G(A,B)| \geq \min ( |\partial(A) \cap \partial(X)|, |\partial(B)\cap \partial(X)|, c) \] 

When $X$ is not connectivity-$c$ linked, a cut $(A,B)$ that violates the linkedness condition is referred to as a {\em violating cut} for $X$. 
A low-connectivity linked set is desirable for us since it can be contracted without changing the connectivity, as formalized in the following lemma.

\begin{lemma}
Let $X$ be a connectivity-$c$ linked set in $G$. Then $G/X$ is a connectivity-$c$ mimicking network for $G$. 
\end{lemma} 

\begin{proof} 
Let $G' = G/X$ be the contracted graph and $v_{X}$ be the contracted vertex in $G'$ that is obtained by contracting $G[X]$. Since we do not contract the terminals, it suffices to show that, for any two subsets $\Xset_A, \Xset_B \subseteq \tset$, we have $\mincut^c_{G'}(\Xset_A, \Xset_B) = \mincut^c_G(\Xset_A, \Xset_B)$. %

Starting with $\mincut^c_{G'}(\Xset_A, \Xset_B) \geq \mincut^c_G(\Xset_A,
\Xset_B)$, we can see that all the edges in $G'$ are also in $G$, which
implies that any cutset in $G'$ is also in $G$. We conclude that the size of
the minimum cut in $G$ must be at most the size of the minimum cut in $G'$,
for any pair of terminals sets. In general, we can say that contraction of
edges only ever increases connectivity, which implies the above.

Let us now show the converse, that is, $\mincut^c_{G'}(\Xset_A, \Xset_B) \leq
\mincut^c_G(\Xset_A, \Xset_B)$. %
Since we are in the unweighted setting, it is sufficient to consider 
$|\Xset_A|, |\Xset_B| \leq c$. %
Suppose that $\mincut^C_{G'}(\Xset_A, \Xset_B) = \ell \leq c$.  
Then there must be $\ell$ disjoint paths connecting $X'_A \subseteq X_A$ to $X'_B \subseteq X_B$ such that $|X'_A| = |X'_B| = \ell$. Denote the set of such paths in $G'$ by $\pset'$. 

We will construct the set of edge-disjoint paths $\pset$ in $G$ connecting $X'_A$ to $X'_B$, thus implying that $\mincut^c_G(X_A, X_B) \geq \ell$. 
We write $\pset'$ as $\pset' = \pset'_1 \cup \pset'_2$ where $\pset'_1$ are the paths that do not go through the contracted vertex $v_{X}$. 
We add the paths in $\pset'_1$ to $\pset$, since they correspond to edge disjoint paths in the original graph $G$.  
For paths in $\pset'_2$, we will need to specify their behavior inside the contracted set $G[X]$. Let $E_{in} \subseteq \partial(X)$ be the set of boundary edges of $X$ that paths in $\pset'_2$ use to enter $v_{X}$; analogously, we define $E_{out} \subseteq \partial(X)$. Notice that $|E_{in}| = |E_{out}| = |\pset'_2| \leq c$. 
Since $X$ is connectivity-$c$ linked, there is a collection of disjoint paths $\pset_{X}$ connecting $E_{in}$ to $E_{out}$.  
We stitch the three parts of the paths in $\pset'_2 \cup \pset_{X}$ together to add to $\pset$: (1) a subpath of some path $P \in \pset'_2$ from a node in $X'_A$ to $E_{in}$, (2) a path in $\pset_{X}$ from such edge in $E_{in}$ to an edge in $E_{out}$, and (3) a subpath of some path $Q \in \pset'_2$ from such an edge in $E_{out}$ to a node in $X'_B$. %
We remark that, even though $\pset$ contains $\ell$ edge-disjoint paths connecting $X'_A$ to $X'_B$, the pairing induced by $\pset$ and $\pset'$ may be different.
\end{proof}

\section{Constructing Bounded Connectivity Mimicking Network}

In this section, we present an algorithm that produces a connectivity-$c$ mimicking network of size at most $3^c \cdot k$.
We show how to efficiently implement our algorithm in the next section.

\subsection{Warmup: Connectivity Two}

Let $G[S] \cup \partial(S)$ be a graph where $|\partial(S)| = k$ and $G[S]$ is
connected. %
In this section, we familiarize readers with the arguments we use by showing how to construct a connectivity-$2$ mimicking network of size at most $2k- 2$.

\begin{lemma}
\label{lem:base case}  
If $k \geq 3$, set $S$ can be decomposed into at most $k-2$ sets that are connectivity-$2$ linked.
\end{lemma}  

As discussed in the previous section, contracting such clusters gives us the desired mimicking network containing $2w -2$ vertices for all $k\geq 3$. 
The proof relies on two simple observations: 

\begin{observation}
Consider a graph $G[S] \cup \partial(S)$. 
Let $X \subseteq S$ be a $2$-connected component in $G[S]$. Then $X$ is connectivity-$2$ linked. 
\end{observation}

\begin{observation}
Consider a graph $G[S] \cup \partial(S)$. 
Let $uv \in E(G)$ where $u, v\in S$ and $deg_G(u) = 2$. Then $\{u,v\}$ is connectivity-$2$ linked. 
\end{observation}

We can now prove the lemma.

\begin{proof}
There are two steps. 
In the first step, we contract all $2$-connected components in $G[S]$ into nodes.   
The graph $G'$ obtained after contraction is a forest. 
Next, whenever there is an edge $uv \in E(G')$ where both $u$ and $v$ are not terminals and $u$ is of degree two, we contract edge $uv$.

We are left with the forest $G''$ such that leafs correspond to terminals, and each internal node has degree at least $3$, except for when it is only adjacent to terminals.
A simple counting argument implies that the number of internal nodes is at most $k-2$: Each leaf receives one token; it sends its token to the parent; each internal node receives at least two tokens from its children, and it keeps one for itself and passes along the rest; in the end, this process leaves at least one token on each internal node and at least $3$ tokens at the root since the root must have degree at least $3$; we conclude that the number of internal nodes is at most $k-2$\footnote{We remark that this claim does not hold when $k=2$. In such case, we have one internal node.}. 

Notice that each internal node in $G''$ corresponds to a connectivity-$2$ linked set $W \subseteq S$. Therefore, this procedure in fact computes a collection of disjoint connectivity-$2$ linked clusters and contracts them. 
\end{proof}

\subsection{General Case}  

In this section, we generalize the arguments above to show that we can decompose $S$ into a relatively small number of connectivity-$q$ linked clusters. 
Our recursive procedure takes set $X \subseteq S$ and connectivity parameter $q$ and is supposed to mark connectivity-$q$ linked clusters inside $X$.  
In particular, the procedure {\sc MarkClusters}($X$, $q$) performs the following steps: 

\begin{itemize}
\item (Base:) If $X$ is connectivity-$q$ linked, mark $X$ as a tentative cluster and return. 
Or if $q = 2$, then we use the procedure described in Lemma~\ref{lem:base case} to mark tentative clusters in $X$ and return.

\item (Inductive:) if $|\partial(X)| \leq 2q-1$, call {\sc MarkClusters}($X$, $q-1$). Else, we find a violating cut $X = A \cup B$ and make recursive calls to {\sc MarkClusters}($A, q$) and {\sc MarkClusters}($B, q$).  

\end{itemize}  

The following observation follows trivially. 

\begin{observation}
When the procedure {\sc MarkClusters}($X,q$) returns, the clusters in $X$ form a partition of $X$. 
\end{observation} 

It is also easy to see that the procedure always terminates. 

\begin{lemma}
Assume that a violating cut can be found in time $f(k,c) {O}(m)$. 
Then {\sc MarkClusters}($X,c$) terminates in time $ O(f(k,c)) {O}(m^2)$. 
\end{lemma}
\begin{proof}
At each recursive call, either the value of $q$ decreases or the number of edges in the induced subgraph of $X$ decreases. The work done outside the recursive calls involves finding a violating cut in $X$, if there exists one, which takes time  $f(k,c)\tilde{O}(|E_G(X,X)|)$. Hence, the total work done at each level of the recursion tree is $f(k,c)\tilde{O}(m)$, since the subsets on which recursion is called at each level are disjoint. This gives an overall runtime of $f(k,c)\tilde{O}(m^2)$.
\end{proof}

Next, we argue that, when the procedure {\sc MarkClusters}($X,q$) returns, we have at most $3^{q} \cdot |X|$ tentative clusters that are contained in $X$ and each such cluster is connectivity-$q$ linked.
After contracting such clusters, we  obtain the desired weak mimicking network. 

\begin{lemma}
Let $W \subseteq X$ be a tentative cluster marked when calling {\sc MarkClusters}($X,q$). Then $W$ is connectivity-$q$ linked. 
\end{lemma}  

\begin{proof}
We prove this by induction. The base case when $X$ is connectivity-$q$ linked is obvious. 
The other base case is when $q=2$, where we can use Lemma~\ref{lem:base case}. 

For the inductive case, there are two possible subcases. The first subcase is when $|\partial(X)| \leq 2q-1$. Let $W \subseteq X$ be a tentative cluster marked by {\sc MarkClusters}($X,q-1$), so $W$ is connectivity-$(q-1)$ linked. We claim that it is also connectivity-$q$ linked.  
Indeed, consider any cut $A \cup B = W$ where $|\partial(A) \cap \partial(W)|
\leq |\partial(B) \cap \partial(W)|$. %
Since $|\partial (X) \leq 2q-1|$, then $|\partial(A) \cap\partial(W)| \leq
q-1$. %
Thus, $(A,B)$ is a violating cut for connectivity-$q$
if and only if it is a violating cut for connectivity-$(q-1)$. The second
subcase follows from definition.
\end{proof} 

\begin{theorem}
The number of tentative clusters (and therefore the size of mimicking network)  is at most $3^c |\tset|$.  
\end{theorem} 

\begin{proof}
We again prove this by induction. 
Let $N_q(k)$ be the maximum number of tentative clusters when running the procedure on $(X,q)$ where $|\partial(X)| = k$. 
We will prove that: 
\begin{align*}
	N_q(k) &\leq 3^{q-2}(k-2(q-1)) & \text{if } k > 2(q-1) \\
	N_q(k) &\leq 3^{q-2}           & \text{if } k \leq 2(q-1)
\end{align*}
The base case when $X$ is already linked is trivial since we have one cluster in $X$. Notice that if $k=1$, the set $X$ is also linked, so we only need to consider the other base case when $k \geq 2$. 
The other base case is when we have $q=2$ and $k \geq 2$. In such case, $N_2(2) = 1$, and when $k \geq 3$, we have that $N_2(k) \leq k-2$, from Lemma~\ref{lem:base case}.  
 
For the inductive step, if $|\partial(X)| \leq 2q-1$, then the procedure simply reduces the value of $q$ to $q-1$, and the induction hypothesis applies. 
Therefore, we consider the case when $|\partial(X)| > 2q-1$, and a violating cut is found, resulting in the calls to {\sc MarkClusters}($A,q$) and {\sc MarkClusters}($B,q$). 
Notice that the boundary edges $\partial(A)$ and $\partial(B)$ are the edges from $\partial(X)$ as well as the edges in the violating cut $E_G(A,B)$.
Assume that $|E_G(A,B)| = \ell$, $|\partial(X) \cap \partial(A)| = k_1$ and $|\partial(X) \cap \partial(B)| = k_2$. Also, assume that $k_1 \leq k_2$. Notice that since this is a violating cut, we have that $k_1 >\ell$.  

There are two possibilities.  
In the first case, suppose that $|\partial(A)| = k_1+\ell > 2(q -1)$. In such case, we have that,
\begin{align*}
	N_q(k) &\leq N_q(k_1+\ell) + N_q(k_2+\ell) \\
	       &\leq 3^{q-2}\big(k_1+\ell-2(q-1)\big) + 3^{q-2}\big(k_2+\ell-2(q-1)\big) \\
	       &\leq 3^{q-2}\big(k_1 + k_2 - 2(q-1) +2\ell-2(q-1)\big) \\
	       &\leq 3^{q-2}\big(k - 2(q-1)\big)
\end{align*}
Otherwise, using the fact that $k_1 >\ell$, 
\begin{align*}
	N_q(k) & \leq N_q(k_1+\ell) + N_q(k_2+\ell) \\
	       &\leq N_q(2q-1) + N_q(k-k_1+\ell) \\
	       &\leq N_{q-1}(2q-1) + N_q(k-1) \\
	       &\leq 3^{q-3}\big(2q-1-2(q-2)\big) + 3^{q-2}\big(k-1-2(q-1)\big) \\
	       &\leq 3^{q-3}\cdot 3 + 3^{q-2}\big(k-1-2(q-1)\big) \\
	       &\leq 3^{q-2}\big(k-2(q-1)\big)
\end{align*}
\end{proof}

\section{Efficiently Finding a Violating Cut}  

In order to speed up our computation, it is sufficient to compute a violating cut efficiently in the subroutine {\sc MarkClusters}($X,q$).  
In this section, we present an algorithm that either finds a violating cut in $X$ or certifies that $X$ is connectivity-$q$ linked. 

Observe first that a violating cut can be found in time $k^{O(c)} \poly(n)$  by simply computing all possible minimum cuts separating any disjoint subset of terminals $\tset_0, \tset_1 \subseteq \tset$ of size $q \leq c$, whose minimum cut contains less than $q$ edges. We will refer to the two sides of the cuts as {\em zero side} and {\em one side} respectively. 
We are, however, aiming at running time $f(c) \poly(n)$, so we cannot afford to do this enumeration to find the ``correct'' $\tset_0$ and $\tset_1$.

Our algorithm will actually solve a more general problem. 
We say that a cut $(A_0, A_1)$ of $G$ is a valid $(Q_0,Q_1,c_0,c_1,\ell)$-constrained cut if 

\begin{itemize}

    \item $Q_0 \subseteq A_0 \setminus \tset$ and $Q_1 \subseteq A_1 \setminus \tset$.
    
    \item $|A_j \cap \tset| \geq c_j$ for $j = 0,1$.  
    
    \item $E_G(A_0,A_1)$ contains at most $\ell$ edges. 
\end{itemize}

In words, $Q_0$ and $Q_1$ are the non-terminals that are ``constrained'' to be on different sides. 
The values of $c_0$ and $c_1$ are the minimum required number of terminals on the sides of $A_0$ and $A_1$ respectively. 

\begin{observation}
Given a sub-routine that finds a valid $(Q_0,Q_1,c_0,c_1,\ell)$-constrained cut in time  given by some function $T(m,k,\max(c_0, c_1, \ell))$,
we can compute a violating cut in $G$ or report that such a cut does not exist in time $O(c T(m,k,c))$.
\end{observation}

In the rest of the section, we shall describe an algorithm that finds a valid $(Q_0,Q_1,c_0,c_1,\ell)$-constrained cut. 
Let $c = \max (c_0,c_1,\ell)$. 
Our algorithm has two steps, encapsulated in the following two lemmas. 

\begin{lemma}[Reduction]
\label{lem:reduction to base} 
There is an algorithm that runs in time $2^{O(c^2)}\cdot k^2 \cdot m$, and reduces the problem of finding a valid $(Q_0,Q_1,c_0, c_1,\ell)$-constrained cut to at most $2^{O(c^2)}$ instances of finding valid $(Q'_0,Q'_1,c'_0,c'_1,\ell')$-constrained cut where $\min(c'_0,c'_1) = 0$. 
\end{lemma}
We remark that each such generated instance may have different constrained parameters. The only property we guarantee is the fact that $\min(c'_0,c'_1) = 0$, that, is, the fact that there is only a one-sided terminal requirement.  

\begin{lemma} [Base case]
\label{lem:algo for base} 
For $\ell \leq c$, 
there is an algorithm that finds a valid $(Q_0,Q_1,0,c,\ell)$-constrained cut (and analogously, $(Q_0,Q_1,c,0,\ell)$-constrained) in time $2^{O(c^2)}\cdot k^2 \cdot m$. 
\end{lemma}

The following theorem follows in a straightforward manner, since every violating
cut is also $(\emptyset, \emptyset, \ell+1, \ell+1, \ell)$-constrained, for some
$\ell \in [c-1]$.

\begin{theorem}
There is an algorithm that runs in time $2^{O(c^2)}\cdot k^2 \cdot m$ and either returns a violating cut or reports that such a cut does not exist. 
\end{theorem}

\subsection{The reduction to the base case} 

In this subsection, we prove Lemma~\ref{lem:reduction to base}. 
The main ingredient for doing so is the following lemma. 

\begin{lemma}
\label{lem: one step} 
There is a reduction from $(Q_0,Q_1,c_0,c_1,\ell)$-constrained cut to solving at most $2^{O(c)}$ instances of finding valid $(Q'_0,Q'_1,c'_0,c'_1, \ell')$-constrained cut where $(c'_0 +c'_1) < (c_0 +c_1)$. 
\end{lemma}

In other words, this lemma allows us to reduce the number of required terminals on at least one of the sides by one. 
Applying Lemma~\ref{lem: one step} recursively will allow us to turn an input instance of $(Q_0, Q_1, c_0,c_1,\ell)$ constrained cut into at most $2^{O(c^2)}$ instances of the base problem: This follows from the fact that at every recursive call, the value of at least one of $c_0$ and $c_1$ decreases by at least one. Therefore the depth of the recursion is at most $2c$, and the ``degree'' of the recursion tree is at most $2^{O(c)}$ as guaranteed by the above lemma.

Let $(G,\tset)$ be an input. We now proceed to prove Lemma~\ref{lem: one step}, that is, we show how to compute a $(Q_0,Q_1, c_0,c_1, \ell)$-constrained cut in $(G,\tset)$.

\paragraph{Our algorithm:}Let $(A'_0, A'_1)$ be a minimum cut in $G$ such that $Q_0 \subseteq A'_0$ and $Q_1 \subseteq A'_1$ and each side contains at least one terminal.  This cut can be found by a standard minimum $s$-$t$ cut algorithm. Observe that the value of this cut is at most $\ell$ if there is a valid constrained cut.  

Such a cut can be used for our recursive approach to solve smaller sub-problems by recursing on $G[A'_i]$ as follows. 
Denote by $\tset_i = A'_i \cap \tset$ for $i=0,1$. 
By definition, each set $\tset_i$ is non-empty, and this is crucial for us. 

If $|E_G(A'_0, A'_1)| > \ell$, the procedure terminates and reports no valid solution. Or, if $|\tset_i| \geq k_i$ for all $i=0,1$, then we have found our desired constrained cut. 
Otherwise, assume that $|\tset_0| < k_0$ (the other case is symmetric). We create a collection of $2^{O(c)}$ instances of smaller sub-problems as follows.

\begin{center}
\fbox{\begin{minipage}{\textwidth}
\noindent {\bf Sub-Instances.}
\begin{itemize}

    \item First, we guess the ``correct'' way to partition terminals in $\tset_0$ into $\tset_0 = \tset_0^0 \cup \tset_0^1$. 
    There are at most $2^c$ possible guesses. 
    
    \item Second, we guess the ``correct'' partition of the (non-terminal) boundary vertices in $V(E_G(A'_0, A'_1)) - \tset$, into $B_0 \cup B_1$ where $B_0$ and $B_1$ are the vertices supposed to be on the zero-side and one-side respectively. 
    Let $\tilde E = E_G(B_0, B_1)$. 
    There are $2^{c}$ possible guesses.

\end{itemize}

\end{minipage}}    
\end{center}

Now we will solve sub-problems in $G[A'_0]$ and $G[A'_1]$. Notice that $G[A'_0]$ has small number of terminals, so we could solve it by brute force. For $G[A'_1]$ we will solve it recursively. 

Let $E_0$ be the minimum cut in $G[A'_0]$ that separates $S_0 = Q_0 \cup (B_0 \cap A'_0) \cup \tset_0^0$ and $T_0 = (B_1 \cap A'_0) \cup \tset_0^1$.   
Next, we solve an instance of valid $(Q'_0, Q'_1, c'_0, c'_1, \ell')$-constrained cut  in $G[A'_1]$ with terminal set $\tset_1$, where $Q'_0 = (B_0 \cap A'_1)$,  $Q'_1 = Q_1\cup (B_1 \cap A'_1)$,  $c'_0 = \max(c_0 - |\tset_0^0|, 0)$, $c'_1 = \max(c_1- |\tset_0^1|,0)$, and $\ell' = \ell - |\tilde E| - |E_0|$. Let $E_1$ be a $(Q'_0, Q'_1, c'_0, c'_1, \ell')$-constrained cut. Our algorithm outputs $E_0 \cup E_1 \cup \tilde{E}$.

\paragraph{Analysis.} 
Clearly, $c'_0 +c'_1 < c_0 +c_1$. 
The following lemma will finish the proof.

\begin{lemma}
There is a $(Q_0,Q_1,c_0, c_1, \ell)$-constrained cut in $(G,\tset)$ if and only if there exist correct guesses $(B_0,B_1, \tset_0^0, \tset_0^1)$ such that a $(Q'_0, Q'_1,c'_0,c'_1, \ell')$-constrained cut exists in $(G[A'_1], \tset_1)$. 
\end{lemma}

\begin{proof}
We argue  the ``if'' part. Suppose that there exists such a guess $(B_0,B_1, \tset_0^0, \tset_0^1)$. 
We claim that $E_0 \cup E_1 \cup \tilde{E}$ is actually a $(Q_0,Q_1,c_0,c_1,\ell)$-constrained cut that we are looking for. Observe that the size of the cut is at most $\ell$. 

We argue that there are two subsets of terminals $\widetilde{\tset}_0$ of size $c_0$  and $\widetilde{\tset}_1$ of size $c_1$ that are separated after removing $E_0 \cup E_1 \cup \tilde{E}$. 
Let $\tset_1^0$ and $\tset_1^1$ be the sets of terminals in $\tset_1$ that are on the side of $Q'_0$ and $Q'_1$ respectively (in particular, $\tset_1^0$ cannot reach $Q'_1$  in $G[A'_1]$ after removing $E_1$). 
Notice that $| \tset_0^0 \cup \tset_1^0| \geq c_0$ and $|\tset_0^1 \cup \tset_1^1| \geq c_1$. The following claim completes the proof of the ``if'' part. 
\end{proof}
\begin{claim}
$Q_0 \cup \tset_0^0 \cup \tset_1^0$ and $Q_1 \cup \tset_0^1 \cup \tset_1^1$ are not connected in $G$ after removing $\tilde{E} \cup E_0 \cup E_1$. 
\end{claim}
\begin{proof}
Let us consider a path $P$ from $Q_0$ to $Q_1$ in $G$; we view it such that the first vertex starts in $Q_0$ and so on until the last vertex on the path is in $Q_1$.  
Let $u$ be the last vertex the path from the start lies completely in $G[A'_0]$ and $v$ be the first vertex such that the path from $v$ to the end lies completely in $G[A'_1]$. Break path $P$ into $P_1 P_2 P_3$ where $P_1$ is the path from the first vertex to $u$, $P_2$ is the path from $u$ to $v$, and $P_3$ the path from $v$ to the last vertex of $P$ in $Q_1$. 
If $|\{u,v\} \cap B_0|= 1$, we would be done, since $P_2$ contains some edge in $\tilde{E}$. 
So it must be that  (i) $u,v \in B_0$ or (ii) $u,v \in B_1$. In case (i), we have $v \in Q'_0$ while the last vertex of $P$ is in $Q_1 \subseteq Q'_1$, so path $P_3$ is path in $G[A'_1]$ connecting $Q'_0$ to $Q'_1$. Hence,  $P_3$ contains an edge in $E_1$. 
In case (ii), we have that $u \in T_0$, while the first vertex in $P$ is in $Q_0 \subseteq S_0$. Therefore, path $P_1$ is a path in $G[A'_0]$ connecting $S_0$ to $T_0$, which must be cut by $E_0$. 

Similar analysis can be done when considering path $P$ that connects $Q_0$ and $\tset_1^1$, or between $\tset_0^0$ and $Q_1 \cup \tset_1^1$. 
The only (somewhat) different case is when path $P$ connects $Q_0$ to $\tset^1_0$. Assume that $P$ is not completely contained in $G[A'_0]$, otherwise it's trivial. 
Let $u$ be the last vertex on $P$ such that the path from the start to $u$ lies completely inside $G[A'_0]$; and $v$ be the first vertex on $P$ such that the path from $v$ to the end of $P$ lies completely inside $G[A'_0]$. 
Again, we break $P$ into three subpaths $P_1 P_2 P_3$ similarly to before. If $u \in B_1$, we are done because $P_1$ would then contain an edge in $E_0$; or if $v \in B_0$, we are also done since $P_3$ would contain an edge in $E_0$; therefore, $u \in B_0$ and $v \in B_1$, so $P_2$ must contain an edge in $\tilde{E}$. 
\end{proof}

To prove the ``only if'' part, assume that $(A_0,A_1)$ is a valid $(Q_0,Q_1, c_0,c_1, \ell)$-constrained cut. 
We argue that there is a choice of guess such that the sub-problem also finds a valid $(Q'_0,Q'_1,c'_0,c'_1,\ell')$-constrained cut. 
We define $B_i = V(E_G(A_0,A_1)) \cap A_i$ for $i=0,1$, and $\tset_0^i = \tset_0 \cap A_i$ for $i=0,1$. With these choices, we have determined the values of $Q'_0$, $Q'_1$, $c'_0$ and $c'_1$. 
The following claim will finish the proof. 

\begin{claim}
There exists a cut $E_0$ that separates $S_0$ and $T_0$ in $G[A'_0]$ and a cut $E_1$ that is a  $(Q'_0, Q'_1, c'_0, c'_1, \ell')$-constrained cut. 
\end{claim}

\begin{proof}
First, we remark that $|E_G(A_0, A_1)| \leq \ell$ and 
\[E_G(A_0, A_1) = E_G(B_0, B_1) \cup E_G(A'_0 \cap A_0, A'_0 \cap A_1) \cup E_G(A'_1 \cap A_0, A'_1 \cap A_1)\] 
To complete the proof of the claim, it suffices to show that  $E_G(A'_0 \cap A_0, A'_0 \cap A_1)$ is an $(S_0,T_0)$ cut in $G[A'_0]$ and that $E_G(A'_1 \cap A_0, A'_1 \cap A_1)$ is a valid constrained cut in $G[A'_1]$. 

The first claim is simple: Since $S_0 \subseteq A_0$ and $T_0 \subseteq A_1$, any path from $S_0$ to $T_0$ in $G[A'_0]$ must contain an edge in $E_G(A'_0 \cap A_0, A'_0 \cap A_1)$. 

The second claim is also simple: (i) $Q'_0 \subseteq A_0$ and $Q'_1 \subseteq A_1$, so the edge set $E_G(A'_1 \cap A_0, A'_1 \cap A_1)$ separates $Q'_0$ and $Q'_1$, (ii) For $i=0,1$, the number of terminals on the $Q'_i$-side must be at least $c_i - |\tset_0^i|$, for otherwise this would contradict the fact that $(A_0,A_1)$ is a $(Q_0,Q_1,c_0,c_1,\ell)$-constrained cut.  
\end{proof}

\begin{lemma}
	\label{lem:runtime}
	Let $c = \max\{\ell, c_0, c_1\}$. Then, the algorithm to reduce the problem of finding a $(Q_0,Q_1,c_0,c_1,\ell)$-constrained cut with $\min{c_0, c_1} > 0$ to an instance to find a $(Q'_0,Q'_1,c'_0,c'_1,\ell)$-constrained cut with $\min{c_1, c_0} =0$ terminates in time $2^{O(c^2)}\cdot k^2\cdot {O}(m)$.
\end{lemma}

\begin{proof}
	Lemma~\ref{lem: one step} implies that the depth of the recursion tree is at the most $2c$ and that each recursive step reduces to solving $2^{O(c)}$ sub-instances. Hence, the total number of nodes in the recursion tree is $2^{O(c^2)}$. 
	
	 The total runtime outside the recursive calls is dominated by a  minimum $s-t$-cut computation. However, we observe that we are only interested in minimum cuts that are of value at the most $c$. Hence, such a cut can be found in time $O(mc)$ using any standard augmentation path based algorithm.
	 Also, recall that we are looking for cuts that have at least one terminal on each side and hence we need to make $k^2$ guesses. The total runtime for this procedure is $k^2\cdot {O}(mc)$ and we have the lemma.
\end{proof}

\subsection{Handling the base case} 
In this subsection, we prove Lemma~\ref{lem:algo for base}, i.e.\ we present
an algorithm that finds a $(Q_0,Q_1,c_0,0,\ell)$-constrained cut $(A'_0,A'_1)$. %
We first consider the case of $c_0=0$: since neither side of the cut must
contain any terminals, we can simply compute a minimum-cut between $Q_0$ and
$Q_1$. %
If one of these is empty (say $Q_1$), we take $A'_0 = V(G)$, $A'_1 =
\emptyset$. %
In any case, let $E_1$ be the edges of the cut. %
Now, there are two possibilities: if $|E_1| \leq \ell$, our cut is a solution
to the subproblem; %
if $|E_1| > \ell$, then there is no cut separating $Q_0$ from $Q_1$ with at
most $\ell$ edges, and therefore, there is no valid constrained cut.

We can now focus on the case where $c_0 >0$. 
We can further assume that $|\tset| \geq c_0$; otherwise, there is no feasible
solution. %
For simplification, we also assume that $Q_0$ is connected; if it is not, we
can add fake edges to make it connected in the run of the algorithm, which we
can remove afterwards (these edges will never be cut, since $Q_0 \subseteq A'_0$). %

\paragraph{Important Cuts.}
The main tool we will be using is the notion of important cuts, introduced by
Marx \cite{Marx06} (see \cite{CyganFKLMPPS15} and references within for other
results using this concept).

\begin{definition}[Important cut]
Let $G$ be a graph and $X,Y \subseteq V(G)$ be disjoint subsets of vertices of
$G$.

A cut $(S_X, S_Y)$, $X \subseteq S_X$, $Y \subseteq S_Y$ is an important cut
if it has (inclusion-wise) maximal reachability (from $X$) among all cuts with
at most as many edges. %
In other words, there is no cut $(S'_X, S'_Y)$, $X \subseteq S'_X$, $Y
\subseteq S'_Y$, such that $|E(S'_X, S'_Y)| \leq |E(S_X, S_Y)|$ and $S_X
\subsetneq S'_X$. %
\end{definition}

\begin{proposition}[{\cite{CyganFKLMPPS15}}]
\label{prop:impcut:exist}
Let $G$ be an undirected graph and $X, Y \subseteq V(G)$ two disjoint sets of vertices. %

Let $(S_X, S_Y)$ be an $(X,Y)$-cut. Then there is an important $(X,Y)$-cut $(S'_X, S'_Y)$ (possibly $S_X = S'_X$) such that $S_X \subseteq S'_X$ and $|E(S'_X, S'_Y)| \leq |E(S_X, S_Y)|$.
\end{proposition}

\begin{theorem}[{\cite{CyganFKLMPPS15}}]
Let $G$ be an undirected graph, $X, Y \subseteq V(G)$ be two disjoint sets of
vertices and $k\geq 0$ be an integer. %
There are at most $4^k$ important $(X,Y)$-cuts of size at most $k$.
\end{theorem}

\begin{proposition}
\label{prop:violcut:impcutforall}
Let $G$ be an undirected graph and $X, Y \subseteq V(G)$ two disjoint sets of
vertices, and let $(S_X, S_Y)$ be an important $(X,Y)$-cut. %

Then $(S_X, S_Y)$ is also an important $(X',Y)$-cut for all $X' \subseteq S_X$.
\end{proposition}

\begin{proof}
Assume that the statement is false for contradiction. Then there is an
important cut $(S'_X, S'_Y)$ for $(X', Y)$, with $|E(S'_X, S'_Y)| \leq |E(S_X,
S_Y)|$ and $S_X \subsetneq S'_X$ by Proposition \ref{prop:impcut:exist}. %
But then, $X \subseteq S_X \subseteq S'_X$, which means $(S'_X, S'_Y)$ is an
$(X,Y)$-cut, and therefore $(S_X, S_Y)$ is not an important cut for $(X,Y)$,
which is a contradiction.
\end{proof}

\paragraph{Cut profile vectors.}
In order to make the exposition of the algorithm clearer, we introduce the concept of cut profile vectors.

\begin{definition}
\label{def:k0:slot}
Let $c, \ell \geq 0$. A \emph{cut profile vector} is a vector of $\lambda \leq c$
pairs of numbers $\{(\kappa_i, \ell_i)\}_{i \in [\lambda]}$, with $\kappa_i \in
[c-1]$, $\ell_i \in [\ell]$, satisfying
\[
c \leq \sum_{i = 1}^\lambda \kappa_i \leq 2c, \quad \quad \sum_{i=1}^\lambda \ell_i \leq \ell
\]
Each of the pairs $(\kappa_i, \ell_i)$ is called a \emph{slot} of this profile. %
We say a cut $(A,B)$ is compatible with a slot $(\kappa_i, \ell_i)$ if $|A
\cap \tset| = \kappa_i$ and $|E(A,B)| = \ell_i$
\end{definition}

\begin{observation}
There are at most $c^c \cdot \ell^c$ different cut profile vectors. 
\end{observation}

Given a cut $(A,B)$, a cut profile vector represents the bounds for terminals
covered and cut edges for each of the components of $G[A]$: there are $\lambda$
connected components, and component $C_i$ contains $\kappa_i$ terminals and
has $\ell_i$ cut edges. %
Our algorithm will enumerate all the possible cut profile vectors and, for each of
them, try to find a solution that fits the constraints given by the input. %
If there is a solution to the problem, there must be a corresponding profile 
vector, and therefore the algorithm finds a solution. %
We refer to Figure \ref{alg:CC:k0} for a formal description of the algorithm.

We will now show that, if there is a solution to the problem, our algorithm
always finds a solution. This implies that, when we output ``No Valid
Solution'', there is no solution. %
From now on, we assume that there is a solution to the problem. %
Let $(A,B)$ be a solution that minimizes the number of connected components of
$G[A]$. %

Let $\cset_0$ be the set of all important cuts $(C, \bar C)$ for $(Q_0, Q_1)$, and let
$\cset$ be the set of all important cuts $(C, \bar C)$ for $(t, Q_1)$, for any $t \in
\tset$.

\begin{lemma}
\label{lem:violcut:impcutcomps}
There is a solution $(A', B')$ such that every connected component $C$ of
$G[A']$ corresponds to an important cut $(C, \bar C)$ in $\cset_0$ or $\cset$. %
Furthermore, the number of connected components of $G[A']$ is not greater than
that of $G[A]$. %
\end{lemma}

\begin{proof}
We will show an iterative process that turns a solution $(A,B)$ into a
solution $(A', B)$ where every component corresponds to an important cut as
above. %

Let $C$ be a component of $G[A]$ that does not correspond to an important cut in
$\cset_0$ or $\cset$. %
Notice that $C$ cannot contain a proper non-empty subset of $Q_0$, since $Q_0
\subseteq A$ and we assume that $Q_0$ is connected. %
If $C$ does not contain any terminals or $Q_0$, we move $C$ to $B$ (resulting
in the cut $(A\setminus C, B \cup C)$). %
Since $C$ is a connected component of $G[A]$, all of the neighbors of $C$ are in
$B$, and therefore moving $C$ to $B$ does not add any cut edges. %

In the remaining case, $C$ contains a terminal $t \in \tset$ or $Q_0$,
but is not an important cut. %
By Proposition \ref{prop:impcut:exist}, there is an important cut $(C',
\bar C')$, with at most as many cut edges as $(C, \bar C)$ and $C \subsetneq
C'$. %
We can replace $C$ by a component corresponding to an important cut by taking
the cut $(A \cup C', B \setminus C')$. %
This is still a valid solution, since all terminals contained in $A$ are
contained in $A \cup C'$, and $Q_0 \subseteq A$, $Q_1 \subseteq B \setminus C'$.
Additionally, the number of edges crossing the cut does not increase: since
$|E(C', \bar C')| \leq |E(C, \bar C)|$, the number of edges added to the
cutset is at most the number of edges removed.

We can apply the operations above until the constraints in the lemma are
satisfied. %
Notice that when applying the operations above, the number of components of
$G[A]$ never increases and the number of vertices in $A$ connected to
terminals in $G[A]$ never decreases. %
Furthermore, each operation changes at least one of the two quantities above,
so this process must finish after a finite number of operations.
\end{proof}

Due to Lemma \ref{lem:violcut:impcutcomps}, we can assume that every connected
component of $G[A]$ corresponds to an important cut. %
Now, let $C^*_0, C^*_1, \ldots, C^*_\lambda$ be the connected components of
$G[A]$, with $C^*_0 \in \cset_0$ being the component that contains $Q_0$. %
Let $\set{(\kappa_i, \ell_i)}_\lambda$ be the cut profile vector corresponding
to the cuts $(C^*_i, \bar C^*_i)$ for $i \in \set{1, \ldots, \lambda}$
(excluding $C^*_0$), meaning that $\kappa_i$, $\ell_i$ are the number of
terminals in $C^*_i$ and the number of edges in the cutset, $E(C^*_i, \bar
C^*_i)$, respectively. %
Notice that, if $C^*_0$ or $C^*_0 \cup C^*_i$ (for some $i \in [\lambda]$)
contain at least $c$ terminals, then we can remove all the other components of
$A$. %
In this case, the algorithm finds $C_0\in\cset_0$ or $C_0 \in \cset_0$, $C_1
\in \cset$ by enumeration and returns a valid solution. %
Otherwise, all the components contain at most $c-1$ terminals each (and thus
$A$ induces a slot vector as in Definition \ref{def:k0:slot}).

Consider the iteration of the algorithm in which the cut profile
vector defined above is considered and $C_0 = C^*_0$. %
The next part of the algorithm (Lines
\ref{alg:k0:forround}--\ref{alg:k0:forroundend}) greedily fills the slots with
compatible important cuts from $\cset$, while making sure that each set
contains a disjoint set of terminals from the others. %
Though it seems that our goal at this stage is to obtain a feasible solution,
what we intend is to obtain a set of terminals, denoted $S$, such that the set
of important cuts for terminals in $S$ contains a feasible solution. %
For instance, if $S$ contains at least one terminal from each $C^*_i$,
$i\in[\lambda]$, our goal is achieved.

The above considerations motivate the following definition. %
We say a slot $i$ is \emph{hit} by $S$ if $S \cap C^*_i \neq \emptyset$. %
Notice that slot $i$ is hit by $S$ if $C_{ji} = C^*_i$ for some $j$, since
the terminals in $C^*_i$ is added to $S$. %
Slot $i$ is also hit by $S$ if, for some $j$, we cannot find a set $C_{ji}$,
since that implies that $C_{ji} = C^*_i$ is not a valid choice, and thus $S
\cap C^*_i \neq \emptyset$. %
Furthermore, if slot $i$ is not hit by $S$, then $C_{ji}$ is found in all
$(c+1)$ rounds. %

Let $\cset_S \subseteq \cset$ be the subset of important cuts containing
terminals in $S$ %
(by Proposition \ref{prop:violcut:impcutforall} these are the important $(t,
Q_1)$-cuts for $t \in S$). %
It is now suficient to show that there is a sequence of $\lambda$ cuts
$\set{(C_i,\bar C_i)}_{i \in \lambda}$ from $\cset_S$, such that all $C_i$
contain disjoint sets of terminals (also disjoint with the terminals in
$C_0$), and such that $(C_i,\bar C_i)$ is compatible with $(\kappa_i,
\ell_i)$. %
Taking $C=C_0 \cup \bigcup_{i=1}^\lambda C_i$, we obtain a feasible solution
$(C, \bar C)$, which may be different from $(A,B)$, but has the same number of
connected components as $G[A]$, and the same numbers of terminals
contained in each component and cut edges separating each component from the
other side of the cut. %
Since the algorithm enumerates all such sequences of $\lambda$ sets, it will
find either $(C, \bar C)$ or a different feasible solution. %

We now define the sets $C_i$: if a slot $i$ is hit by $S$ %
we can set $C_i = C^*_i$, since there is $t \in C^*_i \cap S$, and therefore,
$(C^*_i, \bar C^*_i) \in \cset_S$. %
This cut is trivially compatible with $(\kappa_i, \ell_i)$, and is disjoint to
all other sets defined similarly. 
Let $I_H$ be the set of all $i \in [\lambda]$ such that slot $i$ is hit by $S$,
and let $C^* = \bigcup \set{ C_i: i \in I_H}$. %
All that is left to prove is that, for every slot $i$ that is not hit by $S$,
there is an important cut $(C_i, \bar C_i) \in \cset_S$, which
contains terminals not in any previous $C_{i'}$, $i' \leq i$, or in $C^*$. %
Notice that we have covered at most $c$ terminals so far %
(if we covered more, then the components so far are sufficient and therefore the
number of components of $G[A]$ is not minimal). %
Since there are $c+1$ important cuts $(C_{ji}, \bar C_{ji})$, $j\in[c+1]$, all
compatible with slot $i$ and containing disjoint sets of terminals %
(since the terminals of $C_{ji}$ are added to $S$ after being picked), %
there must be one set $C_{ji}$ that does not contain any of the at most $c$
terminals in $\bigcup_{i' < i} C_{i'}$, or in $C^*$, and we can set
$C_i=C_{ji}$. %
Therefore, a sequence $\set{(C_i,\bar C_i)}_{i \in \lambda}$ exists, and the
algorithm outputs a feasible solution. %

\newcommand{\DeclareProblem}[2]{\DeclareRobustCommand{#1}[0]{\textsf{#2}\xspace}}
\DeclareProblem{\SNDP}{rSNDP}

\section{Survivable Network Design on Bounded-Treewidth}

In this section, we consider the rooted survivable network design problem
(\SNDP), in which we are given a graph $G$ with edge-costs $w$, as well as $h$
demands $(v_i, d_i) \in V \times \Z$, $i \in [h]$, and a root $r \in V$. %
The goal is to find a minimum-cost subgraph that contains, for every demand
$(v_i, d_i)$, $i \in [h]$, $d_i$ edge-disjoint paths connecting $r$ to $v_i$. %

We will show how to solve \SNDP optimally in running time $f(c, \tw(G)) n$, where
$c = \max_i d_i$ is the maximum demand, and $\tw(G)$ is the treewidth of $G$. %
Our algorithm uses the ideas of Chalermsook et al.~\cite{ChalermsookDELV18}
together with connectivity-$c$ mimicking networks.
Our running time is $\exp(9^c\tw(G)^2) m$ which is double-exponential in $c$, but only single-exponential in $\tw(G)$ (whereas the result by Chalermsook et al.~\cite{ChalermsookDELV18} is double-exponential in both $c$ and $\tw(G)$).

Let $(T,X)$ be a tree decomposition of $G$ satisfying the following
properties (see~\cite{ChalermsookDELV18}):
\begin{inparaenum}[(i)]
\item $T$ has height $O(\log n)$;
\item $|X_t| \leq O(\tw(G))$ for all $t \in T$;
\item every leaf bag contains no edges ($E_t = \emptyset$ for all leaves $t \in T$); 
\item every non-leaf has exactly $2$ children.
\end{inparaenum}
Additionally, we add the root $r$ to every bag $X_t$, $t \in \tset$.

The main idea of our algorithm is to assign, to each $t \in T$, a state
representing the connectivity of the solution restricted to $X_t$. %
By assigning these states in a manner that is consistent across $T$, we
can piece together the solution by looking at the states for each individual
node. %
We will show that representing connectivity by two connectivity-$c$ mimicking
networks is sufficient for our purposes, and that we can achieve consistency
across $T$ by using very simple local rules between the state for a node
$t$ and the states for its children $t_1$, $t_2$. %
These rules can be applied using dynamic programming to compute the optimum solution.

\subsection{Local Connectivity Rules}

In this section, we will introduce the local connectivity rules which will
allow us to assign states in a consistent manner to the nodes of $T$. %
The states we will consider consist of two connectivity-$c$ mimicking networks
roughly corresponding to the connectivity of the solution in $E(G_t)$ and $E
\setminus E(G_t)$. %
We then present some rules that make these states consistent across $T$,
while only being enforced for a node and its children.

We remark that this notation deviates from the one used by Chalermsook et
al.~\cite{ChalermsookDELV18}, in which states represent connectivity in
$E(G_t)$ and $E$. %
We do so because taking the union of overlapping mimicking networks would lead
to overcounting of the number of edge-disjoint paths. %

The following local definition of connectivity introduces the desired
consistency rules that we can use to define a dynamic program for the problem. %
Lemma \ref{def:kgst:weakmimnets:unify} shows that a collection of mimicking
networks satisfy the local definition if and only if they represent the
connectivity in $G$ with terminals given by a bag.


\begin{definition}[Local Connectivity]
\label{def:kgst:weakmimnets:unify}

We say that the pairs of weak-mimicking networks $\{(\Gnet_t, \Dnet_t)\}_{t\in V(T)}$ satisfy the
\emph{local connectivity definition} if
\begin{align*}
\Gnet_t &\equiv^c_{X_t} (X_t, \emptyset)     & \text{for every leaf node $t$ of $T$} \\
\Dnet_{\rootn(T)} &\equiv^c_{X_t} (X_t, \emptyset)
\end{align*}
and for every internal node $t \in V(T)$ with children $t_1$ and $t_2$,
\begin{align*}
\Gnet_t &\equiv^c_{X_t} (X_t, E_t) \cup \Gnet_{t_1} \cup \Gnet_{t_2} \\
\Dnet_{t_1} &\equiv^c_{X_t} (X_t, E_t) \cup \Gnet_{t_2} \cup \Dnet_{t} 
\end{align*}
where $A \equiv^c_{X_t} B$ means that $\mincut^c_{A}(S_1, S_2) = \mincut^c_{B}
(S_1, S_2)$ for all disjoint sets $S_1, S_2 \subseteq X_t$. %
\end{definition}

\begin{lemma}
\label{lem:kgst:weakmimnets:unify}
Let $G=(V,E)$ be a graph, and $(\Tcal, X)$ its tree decomposition satisfying
[the usual properties]. %
For every $t \in V(T)$, let $(\Gnet_t, \Dnet_t)$ be a
pair as in Definition \ref{def:kgst:weakmimnets:unify}. %

Then, the pairs $\set{(\Gnet_t, \Dnet_t)}_{t \in T}$ satisfy the local
definitions iff for every $t \in V(\Tcal)$,
\begin{align*}
\Gnet_t &\equiv^c_{X_t} G_t \\
\Dnet_t &\equiv^c_{X_t} G \setminus E(G_t)
\end{align*}
where $A \equiv^c_{X_t} B$ means that $\mincut^c_{A}(S_1, S_2) = \mincut^c_{B}
(S_1, S_2)$ for all disjoint sets $S_1, S_2 \subseteq X_t$. %

\end{lemma}
\begin{proof}
We start by proving the statement for $\Gnet$ by bottom-up induction, and then
the one for $\Dnet$ by top-down induction. %
We will show that 

Let $t \in \Tcal$ be a leaf of the tree decomposition. Then $E(G_t) =
\emptyset$, so the statement immediately follows. Consider now an
internal node $t$ with children $t_1$, $t_2$, and assume
that the claim follows for $t_1$, $t_2$. %
We will define $H'_t = (X_t, E_t) \cup \Gnet_{t_1} \cup \Gnet_{t_2}$, and prove that
$H'_t \equiv^c_{X_t} G_t$. %
That implies that $\Gnet_t \equiv^c_{X_t} H'_t$ (that is, $\Gnet_t$ satisfies
the local connectivity definition) if and only if $\Gnet_t \equiv^c_{X_t}
G_t$. %

Let $S_1, S_2 \subseteq X_t$, and $F$ be the cutset for a mincut between $S_1$
and $S_2$ in $E(G_t)$. %
We will use $c_G(S_1,S_2) = \mincut^c_{G}(S_1, S_2)$ for
conciseness (in this proof only). %
Then
\begin{align*}
c_{G_t}&(S_1, S_2) \\
&= \min(c, |F|) \\
&= \min(c, |F \cap E_t| + |F \cap E(G_{t_1})| + |F \cap E(G_{t_2})|) \\
&\geq \min\parenbig{c, c_{E_t}(S_1, S_2) + c_{E(G_{t_1})}(S_1 \cap X_{t_1}, S_2 \cap X_{t_1}}
			+ c_{E(G_{t_2})}(S_1 \cap X_{t_2}, S_2 \cap X_{t_2})) \\
&= \min\parenbig{c, c_{E_t}(S_1, S_2) + c_{\Gnet_{t_1}}(S_1 \cap X_{t_1}, S_2 \cap X_{t_1}} 
			+ c_{\Gnet_{t_2}}(S_1 \cap X_{t_2}, S_2 \cap X_{t_2})) \\
&\geq c_{H'_t}(S_1, S_2)
\end{align*}

The third inequality follows because each of the three terms corresponds to a
min-cut between $S_1$ and $S_2$ for the respective edge sets. %
The fourth inequality follows by induction hypothesis, and the final one
follows by definition of $H'_t$. %
For this last step, we crucially use that $X_{t_1} \cap X_{t_2} \subseteq
X_t$, which means that any cut for $E_t$, $\Gnet_{t_1}$ and $\Gnet_{t_2}$ uses
disjoint edges and disjoint vertices outside of $X_t$. These edges provide an
upper bound for the cut $c_{H'_t}$.

Analogously, we can prove that $c_{G_t} \leq c_{H'_t}$, by
taking a set of edges $F'$ of $H'_t$ that realizes the minimum cut in that
graph. The same steps then apply to prove the desired inequality. 
\begin{align*}
c_{H'_t}(S_1, S_2) 
&= \min(c, |F'|) \\
&= \min(c, |F' \cap E_t| + |F' \cap E(\Gnet_{t_1})| + |F' \cap E(\Gnet_{t_2})|) \\
&\geq \min\parenbig{c, c_{E_t}(S_1, S_2) + c_{\Gnet_{t_1}}(S_1 \cap X_{t_1}, S_2 \cap X_{t_1}}
			+ c_{\Gnet_{t_2}}(S_1 \cap X_{t_2}, S_2 \cap X_{t_2})) \\
&= \min\parenbig{c, c_{E_t}(S_1, S_2) + c_{G_{t_1}}(S_1 \cap X_{t_1}, S_2 \cap X_{t_1}} 
			+ c_{G_{t_2}}(S_1 \cap X_{t_2}, S_2 \cap X_{t_2})) \\
&\geq c_{G_t}(S_1, S_2)
\end{align*}
This concludes the first part of the proof. %

For the second part of the proof, we will use top-down induction. %
For $t=r$, notice that $E\setminus E(G_t) = \emptyset$, so the statement
follows. %
We now prove the equality for a node $t_1$ with parent $t$ and sibling $t_2$. %
Let  $H_{t_1} = (X_t, E_t) \cup \Gnet_{t_2} \cup \Dnet_{t}$, and prove that $H_t \equiv^c_{X_t} G \setminus G_t$. %
This implies the statement, as it shows that $\Dnet_t \equiv^c_{X_t} H_t$
(that is, $\Dnet_t$ satisfies the local connectivity definition) if and only
if $\Dnet_t \equiv^c_{X_t} G \setminus G_t$. %

\newcommand{\twolineeq}[3]{
	\begin{aligned}[t]
	#1 & #2 \\
	   & #3
	\end{aligned}
}

Let $S_1, S_2 \subseteq X_{t_1}$, and $F$ be the cutset for a
mincut between $S_1$ and $S_2$ in $E\setminus E(G_{t_1})$. Then
\begin{align*}
c_{E\setminus E(G_{t_1})}(S_1, S_2) 
&= \min(c, |F|) \\
&= \min(c, |F \cap E_t| + |F \cap (E\setminus E(G_t))| + |F \cap E(G_{t_2})|) \\
&\geq 
\twolineeq{\min(c, c_{E_t}(S_1 \cap X_t, S_2 \cap X_t)}{
+ c_{E\setminus E(G_t)}(S_1 \cap X_t, S_2 \cap X_t)}{
+ c_{E(G_{t_2})}(S_1 \cap X_{t_2}, S_2 \cap X_{t_2}))} \\
&= 
\twolineeq{\min(c, c_{E_t}(S_1 \cap X_t, S_2 \cap X_t)}{
+ c_{\Dnet_t}(S_1 \cap X_t, S_2 \cap X_t)}{
+ c_{\Gnet_{t_2}}(S_1 \cap X_{t_2}, S_2 \cap X_{t_2}))} \\
&\geq c_{H_{t_1}}(S_1, S_2)
\end{align*}

Similarly to the proof above, we use the fact that $F \cap E_t$, $F \cap
(E\setminus E(G_t))$, $F \cap E(G_{t_2})$ are cuts in the subgraphs
$E_t$, $G \setminus E(G_t)$, $E(G_{t_2})$ respectively. The last step
follows from the fact that the three terms correspond to cuts in $E_t$, $\Dnet_t$
and $\Gnet_{t_2}$, and therefore their union forms a cut in $\Dnet_t \cup E_t \cup
\Gnet_{t_2}$. Since $H_{t_1} \equiv^c_{X_{t_1}} \Dnet_t \cup E_t \cup \Gnet_{t_2}$,
the inequality follows. The converse follows similarly:
\begin{align*}
c_{H_{t_1}}(S_1, S_2)
&= \min(c, |F|) \\
&= \min(c, |F \cap E_t| + |F \cap E(\Dnet_t)| + |F \cap E(\Gnet_{t_2})|) \\
&\geq 
\twolineeq{\min(c, c_{E_t}(S_1 \cap X_t, S_2 \cap X_t)}{ 
  + c_{\Dnet_t}(S_1 \cap X_t, S_2 \cap X_t)}{
  + c_{\Gnet_{t_2}}(S_1 \cap X_{t_2}, S_2 \cap X_{t_2}))} \\
&= 
\twolineeq{ \min(c, c_{E_t}(S_1 \cap X_t, S_2 \cap X_t)}{
  + c_{E\setminus E(G_t)}(S_1 \cap X_t, S_2 \cap X_t)}{
  + c_{E(G_{t_2})}(S_1 \cap X_{t_2}, S_2 \cap X_{t_2}))} \\
&\geq c_{E\setminus E(G_{t_1})}(S_1, S_2) 
\end{align*}
This completes the proof.
\end{proof}

\subsection{Dynamic Program for \SNDP}

In this section, we present an algorithm for \SNDP on bounded-treewidth
graphs, which uses dynamic programming to compute a solution bottom-up. %
Our goal is to assign two mimicking networks $\Gnet_t$, $\Dnet_t$ to each node
$t \in T$, corresponding to the connectivity of the solution in $E(G_t)$
and $E\setminus E(G_t)$. %
We argue that any solution for $G_t$, $t \in T$ that is compatible with a
state $(\Gnet_t, \Dnet_t)$ can be interchangeably used, which implies that the
dynamic program will obtain the minimum-cost solution. %

We define a dynamic programming table $D$, with entries $D[t, \Gnet, \Dnet]$,
$t \in T$, $\Gnet$, $\Dnet$ connectivity-$c$ mimicking networks with
terminal set $X_t$. %
The entry $D[t, \Gnet, \Dnet]$ represents the minimum cost of a solution $F$
that is consistent with $\Gnet$ (i.e.\ $F \equiv^c_{X_t} \Gnet$), such that $F
\cup \Dnet_t$ satisfies all the demands contained in $G_t$.

We compute $D[t, \Gnet, \Dnet]$ as follows:
\begin{itemize}
	\item For any leaf $t$, set $D[t, \emptyset, \Dnet] = 0$ and $D[t, \Gnet,
	\Dnet] = +\infty$ for $\Gnet \neq \emptyset$;
	\item For the root node $\rootn(T)$, set $D[\rootn(T), \Gnet, \Dnet] = +\infty$ if $\Dnet
	\neq \emptyset$;
	\item For any demand $(v_i, d_i)$, and $t \in T$ such that $v_i \in X_t$, set $D[t, \Gnet, \Dnet] = +\infty$ if $\Gnet \cup \Dnet$ contain fewer than $d_i$ edge-disjoint paths connecting $r$ to $v_i$.
\end{itemize}
For all other entries of $T$, compute it recursively as:
\begin{align*}
D[t, \Gnet, \Dnet] 
= \min \Big\{ w(Y) + D[t_1, \Gnet_1, \Dnet_1] &+ D[t_2, \Gnet_2, \Dnet_2] :
Y \subseteq E_t, \\
&\Gnet \equiv^c_{X_t} Y \cup \Gnet_1 \cup \Gnet_2, \\
&\Dnet_1 \equiv^c_{X_{t_1}} Y \cup \Dnet \cup \Gnet_2, \\
&\Dnet_2 \equiv^c_{X_{t_2}} Y \cup \Dnet \cup \Gnet_1
\Big\}
\end{align*}

We now want to prove that the dynamic program is feasible, that is, that the
entries $D[\rootn(T), \Gnet, \emptyset]$ correspond to feasible solutions;
and that it is optimal, meaning that we will obtain the optimum solution to
the problem. %

To prove that the dynamic program is feasible, notice that, by definition, any
solution obtained induces a choice of $Y_t$, $\Gnet_t$, $\Dnet_t$ for each $t
\in T$. Let $Y = \cup_{t \in T}$. %
The recursion formula of the dynamic program implies that the pairs
$\set{(\Gnet_t, \Dnet_t)}_{t \in T}$ satisfy the local connectivity
definition with regard to the graph $(V, Y)$. %

By Lemma \ref{lem:kgst:weakmimnets:unify}, this implies that 
\begin{align*}
\Gnet_t \equiv^c_{X_t} G_t[Y], 
\Dnet_t \equiv^c_{X_t} G[Y] \setminus E(G_t), 
\end{align*}
and hence, $\Gnet_t \cup \Dnet_t \equiv^c_{X_t} G[Y]$.

Let $(v_i, d_i)$ be a demand and $t \in T$ be a node such that $v_i \in
X_t$. %
Since we know that $\Gnet_t \cup \Dnet_t$ contains $d_i$ edge-disjoint paths
from $r$ to $v_i$ (otherwise $D[t, \Gnet_t, \Dnet_t] = + \infty$), then we
know that the minimum cut separating $r$ from $v_i$ has at least $d_i$ edges,
which implies that $Y$ must also contain $d_i$ edge-disjoint paths connecting
$r$ and $v_i$.

For the converse, we will prove that any feasible solution $F$ can be captured
by the dynamic program. %
Given $F$, it is sufficient to take $\Gnet_t$, $\Dnet_t$ to be
connectivity-$c$ mimicking networks for $G_t[F]$, $G[F] \setminus E(G_t)$,
respectively. %
By Lemma \ref{lem:kgst:weakmimnets:unify} (applied to graph $(V,F)$), we know
that $\set{(\Gnet_t, \Dnet_t)}_{t \in T}$ satisfy the local connectivity
definition for $(V,F)$, and therefore $D[t,\Gnet_t,\Dnet_t]$ can be computed
recursively from $D[t_1, \Gnet_{t_1}, \Dnet_{t_1}]$, $D[t_2, \Gnet_{t_2},
\Dnet_{t_2}]$, $Y_t = F \cap E_t$.

Let $(v_i, d_i)$ be a demand and $t \in T$ be a node such that $v_i \in
X_t$. %
Since $F$ is a feasible solution, it contains $d_i$ edge-disjoint paths from
$r$ to $v_i$, and therefore $\mincut^c_F(\set{r}, \set{v_i}) \geq d_i$. This
implies that $\mincut^c_{\Gnet_t \cup \Dnet_t}(\set{r}, \set{v_i}) \geq d_i$,
and thus $\Gnet_t \cup \Dnet_t$ contains $d_i$ edge-disjoint paths from $r$ to
$v_i$ (and is a valid entry of $T$).

We conclude that the dynamic program above computes an optimum solution for
\SNDP. By Theorem \ref{thm:connc:mimnet}, there is a weak-mimicking network
for any graph, with $w$ terminals, of size $O(3^c c w)$. Since the number of
edges is at most the square of the number of vertices, there are $O(\exp(9^c
c^2 w^2))$ possible states for each node, which completes the proof.

\begin{figure}[t]
\label{alg:CC:k0}
\begin{algorithmic}[1]
\Function{$\CC$}{$G$, $\tset$, $Q_0$, $Q_1$,$c$, $\ell$}
	\If {$c = 0$}
	\label{alg:k0:ifk0}
		\State Compute a min-cut $(A'_0, A'_1)$ such that $Q_0\subseteq A'_0, Q_1\subseteq A'_1$
		\If {$|E(A'_0, A'_1)| \leq \ell$} 
			\State \Return $(A'_0, A'_1)$
		\Else 
			\State\Return \nosolution
		\EndIf
	\EndIf
	\label{alg:k0:ifk0end}
	\Statex

	\State Compute the set $\cset_0$ of all  important $(Q_0, Q_1)$-cuts with at most $\ell$ cut edges
	\State Compute the set $\cset$ of all important $(t, Q_1)$-cuts for $t \in \tset$ with at most $\ell$ cut edges
	\label{alg:k0:impcut}
	\Statex

	\State Find an important cut $(C_0, \bar{C_0}) \in \cset_0$ such that $|C_0 \cap \tset| \geq c$
	\label{alg:k0:spec1}
	\If {$(C_0, \bar{C_0})$ exists}
		\State\Return $(C_0, \bar{C_0})$
	\EndIf
	\State  \parbox[t]{0.9\textwidth}{
			Find important cuts $(C_0, \bar{C_0}) \in \cset_0$, $(C_1, \bar{C_1}) \in \cset$, 
			such that $(C_0 \cap \tset) \cap (C_1 \cap \tset) = \emptyset$, \par
			\hskip\algorithmicindent $|C_0\cap\tset| + |C_1 \cap \tset| \geq c$, 
			and $\card{E(C_0 \cup C_1, \bar{C_0} \cap \bar{C_1})} \leq \ell$
			}
	\If {$(C_0, \bar{C_0})$, $(C_1, \bar{C_1})$ exist}
		\State\Return $(C_0 \cup C_1, \bar{C_0} \cap \bar{C_1})$
	\EndIf
	\label{alg:k0:spec2}

	\Statex
	\For {all cut profile vectors $\set{(\kappa_i, \ell_i)}_\lambda$, and all $(C_0, \bar{C_0}) \in \cset_0$ }
	\label{alg:k0:forslotvec}
		\State $S \gets C_0 \cap \tset$
		\For {$j \in \set{1, \ldots, c+1}$} \InlineComment{Round $j$}
		\label{alg:k0:forround}
			\For {$i \in \set{1, \ldots, \lambda}$}
			\label{alg:k0:forslot}
				\State Find $(C_{ji}, \bar C_{ji}) \in \cset$ compatible with slot $(\kappa_i, \ell_i)$, 
						such that $C_{ji} \cap S = \emptyset$
				\If {$C_{ji}$ exists} 
					\State $S \gets S \cup (C_{ji} \cap \tset)$
				\EndIf
			\EndFor
		\EndFor
		\label{alg:k0:forroundend}

		\Statex
		\State Let $\cset_S = \set{(C, \bar C) \mid C \cap S \neq \emptyset}$
		\State \parbox[t]{0.8\textwidth}{
			Find (by enumeration) $\set{(C_i,\bar C_i)}_{i \in \lambda}$, with $(C_i,\bar C_i) \in \cset_S$ compatible with slot $i$, \par 
			\hskip\algorithmicindent and all sets $C_i\cap \tset$ are disjoint (including $C_0 \cap \tset$)
			}
		\label{alg:k0:enumset}
		\If {$\set{(C_i,\bar C_i)}_{i \in \lambda}$ exists}
			\State Let $C=C_0 \cup \bigcup_{i=1}^\lambda C_i$
			\State \Return $(C, \bar C)$
		\EndIf
	\EndFor

	\Statex
	\State\Return \nosolution
	\InlineComment{No solution found for any cut profile}
\EndFunction
\end{algorithmic}
\caption{Algorithm to find a constrained cut in the base case}
\end{figure}

\paragraph{Acknowledgement:} Parinya Chalermsook has been  supported by European Research Council (ERC) under the European Union’s Horizon 2020 research and innovation programme (grant agreement No 759557) and by Academy of Finland Research Fellows, under grant number 310415.
Bundit Laekhanukit has been partially supported by the 1000-talent award by the Chinese government.

\bibliographystyle{abbrv}
\bibliography{violating-cut}

\end{document}